\documentclass[10pt,conference]{IEEEtran}

\usepackage{graphicx}
\graphicspath{{./}{./figs/}}
\usepackage{algorithm}
\usepackage{algorithmic}
\usepackage{amsmath}
\usepackage{subfig}
\usepackage{booktabs}

\renewcommand\thesection{\arabic{section}}
\renewcommand\thesubsection{\thesection.\arabic{subsection}}
\renewcommand\thesubsubsection{\thesubsection.\arabic{subsubsection}}

\begin{document}

\title{Network Flow-based Simultaneous Retiming and Slack Budgeting for Low Power Design}

\iftrue
\author
{
\IEEEauthorblockN{Bei Yu\IEEEauthorrefmark{1}, Sheqin
Dong\IEEEauthorrefmark{1}, Yuchun Ma\IEEEauthorrefmark{1}, Tao
Lin\IEEEauthorrefmark{1}, Song Chen\IEEEauthorrefmark{2} and Satoshi
Goto\IEEEauthorrefmark{2}}
\IEEEauthorblockA{
\IEEEauthorrefmark{1}Department of Computer Science \& Technology, Tsinghua University, Beijing, China}
\IEEEauthorblockA{\IEEEauthorrefmark{2}Graduate School of IPS, Waseda University, Kitakyushu, Japan}
Email: \{b-yu07@mails, dongsq@mail\}.tsinghua.edu.cn}
\fi

\maketitle


\newtheorem{problem}{Problem}
\newtheorem{definition}{Definition}
\newtheorem{theorem}{Theorem}
\newtheorem{lemma}{Lemma}

\begin{abstract}
Low power design has become one of the most significant requirements when CMOS technology entered the nanometer era.
Therefore, timing budget is often performed to slow down as many components as possible so that timing slacks can
be applied to reduce the power consumption while maintaining the performance of the whole design.
Retiming is a procedure that involves the relocation of flip-flops (FFs) across logic gates to
achieve faster clocking speed. In this paper we show that
the retiming and slack budgeting problem can be formulated to a convex cost dual network flow problem.
Both the theoretical analysis and experimental results show the
efficiency of our approach which can not only reduce power consumption but also speedup previous work.
\end{abstract}

\section{Introduction}
Timing constraint design and low power design have become
significant requirements when the CMOS technology entered the
nanometer era. On the one hand, more and more devices trend to be
put in the small silicon area while at the same time the clock
frequency is pushed even higher. As an effective timing optimization
scheme, retiming is a procedure that involves the relocation of
flip-flops (FFs) across logic gates to achieve faster clocking
period. On the other hand, to tackle the tremendous growth in the
design complexity, timing budgeting is performed to relax the timing
constraints for as many components as possible without violating the
system's timing constraint. Therefore, both retiming and timing
budget might influence the timing distribution of the design
greatly.

Since Leiserson and Saxe proposed the idea of retiming in 1983
\cite{retiming_Algorithmica1991}, it has become one of the most
powerful sequential optimization techniques. In
\cite{retiming_TVLSI98Sachin}, the min-area retiming problem was
solved by min-cost network flow algorithm. Recent publications
\cite{retiming_ASPDAC05Zhou} and \cite{retiming_ACM08Zhou} proposed
a very efficient retiming algorithm for minimal period by algorithm
derivation. \cite{retiming_DAC06Zhou} and \cite{retiming_DAC08Zhou}
respectively presented efficient incremental algorithms for min-period
retiming under setup and hold constraints, and min-area retiming
under given clock period.

For timing-constrained gate-level synthesis, timing slack is an
effective method for circuit's potential performance improvement.
The components with relaxed timing constraints can be further
optimized to improve system's area, power dissipation, or other
design quality metrics. The slack budgeting problem has been studied
well. Some of the previous slack budgeting approaches are suboptimal
heuristics such as Zero-Slack Algorithm (ZSA)
\cite{slack_TCAD89Nair}.
\cite{slack_DAC96Maji}\cite{slack_TCAD02Maji} formulated the slack
budgeting problem as Maximum-Independent-Set (MIS) on sensitive
transitive closure graph. In \cite{slack_ICCAD04Majid} and
\cite{slack_TCAD06Majid}, authors proposed combinatorial methods
based on net flow approach to handle the slack budget problem.

Budgeting problem can be extended to describe exactly real-word
applications, such as gate resizing, multiple $Vdd$
and multiple $Vth$ assignment \cite{slack2_ISLPED03Nguyen}\cite{slack2_DAC04Srivastava}\cite{slack2_ISLPED04Kulkarni}.
Since the number of logically equivalent cells in a library is limited,
it is reasonable to limit the possible slack value in real designs .
In \cite{slack_ISQED08Ma}, Qiu et al. showed that power reduction is
not proportional to the slack amount and propose a piecewise linear
model to approximate the relationship between slack and power
reduction. In this paper, we adopt the same model and consider discrete
slack budgeting problem. Note that our method can be easily
transferred into continuous slack budgeting problem.

\begin{figure} [tb]
  \centering%
  \subfloat[]{%
    \label{fig:FFexample1}
    \includegraphics[width=0.23\textwidth]{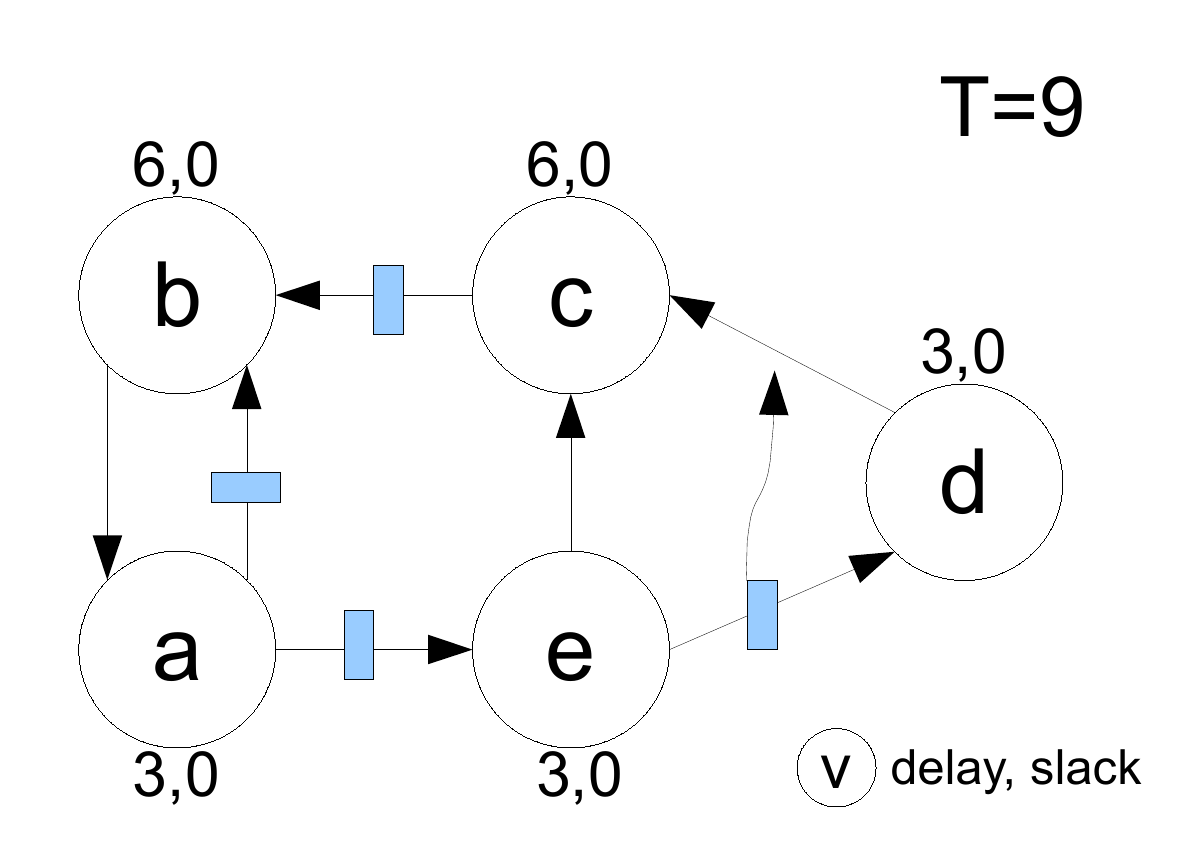}}\hspace{0.1em}%
  \subfloat[]{%
    \label{fig:FFexample2}
    \includegraphics[width=0.23\textwidth]{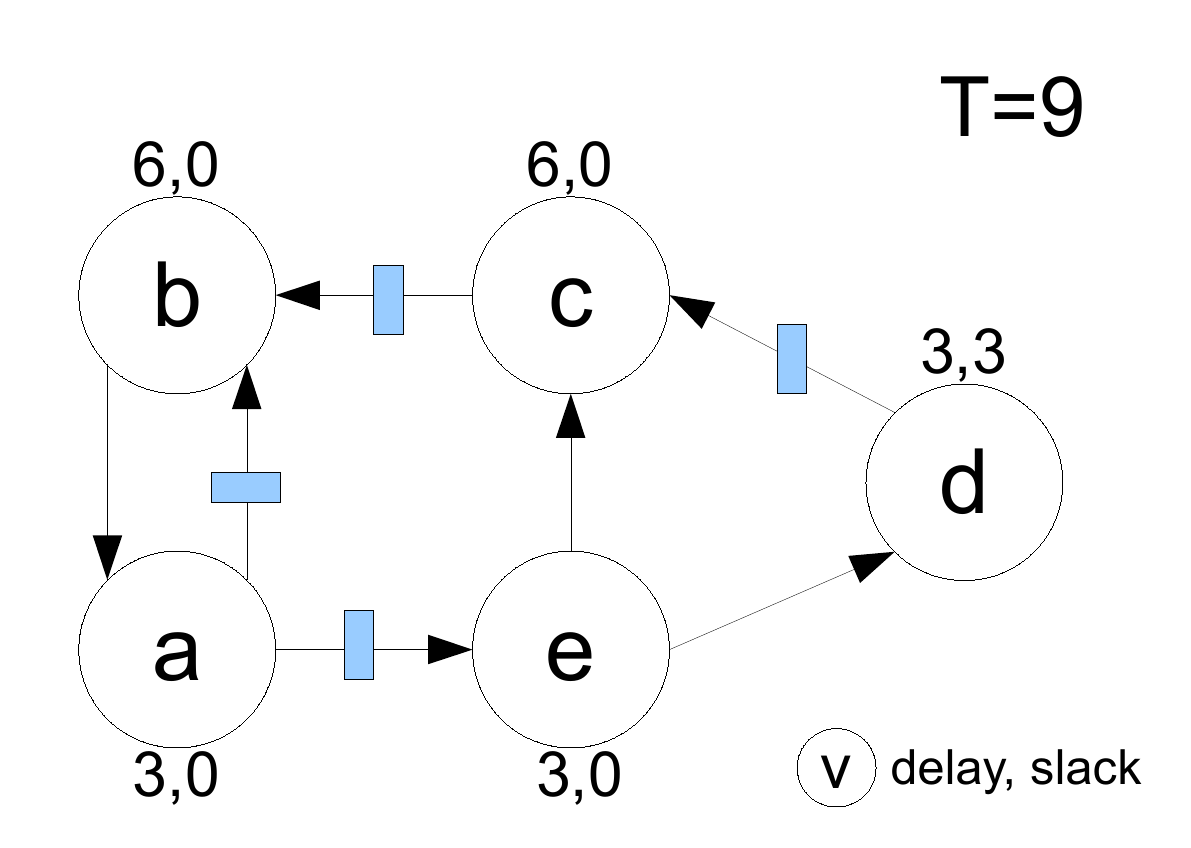}}
  \caption{Relocate FFs to increase potential slack without violating timing constraint. ~(a)No potential slack in this circuit. ~(b)moving the FF from edge $de$ to edge $cd$, the potential slack can
be increased from $0$ to $3$.}
  \label{fig:FFexample}
\end{figure}

Nearly all the existing slack budgeting algorithms are either used
for combinatorial circuit, or limited to fixed FF locations. At the
early design stages, it is flexibility to schedule pipeline or
timing distribution to obtain more timing slack. As shown in Fig.
(\ref{fig:FFexample}), the period of a circuit is minimized with the
delay and slack labeled beside each gate as well. It is seen that
there is no potential slack in this circuit. However, if retiming
and slack budget process is taken, $i.e$. moving the FF from edge
$de$ to edge $cd$, the potential slack can be increased from $0$ to
$3$, keeping the period minimized at the same time.

A simultaneous retiming and slack budgeting algorithm for dual-Vdd
programmable FPGA power reduction was proposed in
\cite{retiming_DAC06Hu}. In \cite{slack_DATE07Zhou} Lin et al. proved that slack budgeting
problem can be viewed as a convex retiming problem. However they
failed to formulate retiming and slack budgeting simultaneously. In \cite{retiming_ASPDAC10Liu} authors
proposed a simultaneously slack budgeting and incremental retiming
algorithm to maximize the potential slack by retiming for
synchronous sequential circuit. They proposed a reasonable
algorithm flow, however, their solution quality suffers in two
aspects. First, there was no guarantee that the algorithm will get
optimal solution because iterative strategy is easily trapped in
local optimum. Besides, the slack budget problem was translated to a
Maximal Independent Set (MIS) problem, which is a NP-hard problem.

\cite{flow_2001Ahuja} showed that for an Integer Linear Programming (ILP)
with separable convex objective functions and special form of constraints,
it can be viewed as convex cost dual network flow problem and solved in polynomial time.
This model has been adopted in various works, such as buffer insertion
\cite{flow_ICCAD05Zhou}, multi-voltage supply \cite{mvs_ICCAD08MA}\cite{mvs_GLSVLSI09Yu},
clock skew scheduling \cite{flow_ASPDAC07Zhou} and slack budgeting \cite{slack_DATE07Zhou}.

In this paper we first formulate retiming and slack budgeting
problem as an Integer Linear Programming (ILP) problem. Since ILP has
been listed to be one of the known NP-hard problems, we then show
how to transform this problem to the convex cost dual
network flow problem with just a little loss of optimality. Experimental results show that our algorithm
can not only reduce power consumption, increase total slack budgeting,  but also effectively speedup previous work.

The remainder of this paper is organized as follows. Section
\ref{sec:problem} defines the simultaneous slack budget and retiming
problem. Section \ref{sec:solution} presents our algorithm flow.
Section \ref{sec:result} reports our experimental results. At last,
Section \ref{sec:conclusion} concludes this paper.

\section{Problem Formulation}
\label{sec:problem}

As shown in \cite{retiming_Algorithmica1991}, we model a synchronous
sequential circuit as a directed graph $G(V, E, d, w)$, each vertex
$i \in V$ represents a combinational gate and each edge $(i, j) \in
E$ represents a signal passing from gate $i$ to $j$. Non-negative gate delays are given as vertex weights $d: V \to R$.
Non-negative integer $w: E \to Z$ as the edge weight represents the number of FFs
on the signal pass. The max clock period is given as $T$.

For each vertex, three non-negative labels, $a_i/\gamma_i/s_i$,
represent the latest arrival time, require time, and slack of vertex
$i$. $a_i$ and $\gamma_i$ can be calculated as follows:
\begin{equation}
    \label{eq:a}
    \left\{
      \begin{array}{lll}
        a_i = d_i & if & w(k, i) > 0 \ or \ i \in PI \\
        a_i = max_j (a_j+d_j) & & \forall j \in FI(i)
      \end{array}
    \right.
\end{equation}
\begin{equation}
    \label{eq:gamma}
    \left\{
      \begin{array}{lll}
        \gamma_i = T & if & w(k, i) > 0 \ or \ i \in PO \\
        \gamma_i = min_j (\gamma_j-d_j) & & \forall j \in FI(i)
      \end{array}
    \right.
\end{equation}
where $PI$ is set of all primary inputs and $PO$ is set of all
primary outputs. $FI(i)$ and $FO(i)$ represent the incoming and
outgoing gates to gate $i$ respectively. Then slack $s_i$ is then
calculated by
\begin{equation}
    \label{eq:s}
    s_i = \gamma_i - a_i
\end{equation}

A retiming of a circuit $G$ is an integer-valued vertex-labeling
$r$, which represent how many FFs are moved from the outgoing edges
to the incoming edges of each vertex. Thus the number of FFs on edge
$(i, j)$ with label $r$ is formulated as follow:
\begin{equation}
    w_{i, j} + r_j - r_i \notag
\end{equation}

\begin{definition}{\textbf{Power Slack Curve}}
- Each gate $i$ is given $k$ discrete slack levels, and the power-slack tradeoff is represented
by $\{(si^1,P(si^1)), \cdots$, $(si^k,P(si^k))\}$. In the Power Slack Curve,
each point is connected to its neighboring point(s) by a linear segment.
\end{definition}

Based on the relationship between power reduction and slack provided by \cite{slack_ISQED08Ma},
we assume Power Slack Curve is a convex decreasing function.

\begin{definition}{\textbf{Simultaneous Slack Budget and Retiming Problem}}
- Given a directed graph $G=(V, E, d, w)$ representing a synchronous
sequential circuit, and period constraint $T$, we want to find FFs
reallocation represented by $r$, such that the power consumption
obtained by slack budgeting is minimized under the period
constraint.
\end{definition}

According to the above definitions and notations, the simultaneous
slack budget and retiming problem can be easily formulated into the
following mathematical program.
\begin{align}\label{eq:ssbr}
    \textrm{min}  & \ \sum_{i\in V} P(s_i) & \tag{$I$}\\
    \textrm{s.t}. & \ \ (\ref{eq:a}) - (\ref{eq:s})          & {} \nonumber\\
                  & r_j - r_i \ge - w_{i, j}         & \forall (i, j) \in E    {} \notag\\
                  & s_i \in \{s_i^1, \cdots, s_i^k\} & \forall i \in V         {} \notag\\
                  & a_i \le T                        & \forall i \in V         {} \notag
\end{align}

\section{Methodology}
\label{sec:solution}
\subsection{MILP Formulation}

The MILP formulation for retiming synchronous circuits is originally
presented in \cite{retiming_Algorithmica1991} to minimize clock
period. The clock period $\Phi(G) \le T$ if and only if there exists
an assignment of real values $a_i$ and an integer value $r_i$ to
each vertex $i \in V$ such that the following conditions are
satisfied:
 \begin{eqnarray}
        a_i \ge d_i + s_i          & \forall i \in V \\
        a_i \le T                    & \forall i \in V  \\
        r_i - r_j \le w_{ij}         & \forall (i, j) \in E \\
        a_j \ge a_i + d_i + s_i & if \ r_i - r_j = w_{ij}
 \end{eqnarray}

Suppose $R_i = r_i + a_i/T$, then $a_i = T \cdot R_i - T \cdot r_i$.
The problem can be formulated as (\ref{eq:model1}).
\begin{align}
    \label{eq:model1}
    \textrm{min}  &     \ \ \ \  \sum_{i\in V} P(\bar s_i)& \tag{$\uppercase\expandafter{\romannumeral2}$}\\
    \textrm{s.t}.\ \ & \bar R_i - \bar r_i \ge \bar s_i            &\forall i \in V        \label{model1:Rr>s}\tag{$IIa$}\\
            & \bar R_i - \bar r_i \le T                            &\forall i \in V        \label{model1:Rr<c}\tag{$IIb$}\\
            & \bar r_j - \bar r_i \ge - T \cdot w_{ij}             &\forall (i, j) \in E   \label{model1:rr<-cw}\tag{$IIc$}\\
            & 0 \le \bar R_i, \bar r_i \le \bar N_{ff}             &\forall i \in V       \label{model1:Rr}\tag{$IId$}\\
            & \bar s_i =\{\bar s_i^1, \cdots, \bar s_i^k\}         &\forall i \in V       \label{model1:si} \tag{$IIe$}\\
            & 0 \le \bar s_i \le T                                 &\forall i \in V       \label{model1:s}\tag{$IIf$}\\
            & \bar R_j - \bar R_i \ge t_{ij}                       &\forall (i, j) \in E    \label{model1:RR>tcw}\tag{$IIg$}\\
            & t_{ij} \ge \bar s_j - T \cdot w_{ij}                 &\forall (i, j)\in E \label{model1:t>s}\tag{$IIh$}
\end{align}
where $\bar N_{ff} = N_{ff} \cdot T$, $\bar s_i = d_i + s_i$, $\bar
r_i = r_i \cdot T$ and $\bar R_i = R_i \cdot T$. For each gate $i$,
$\bar s_i^j = s_i^j+d_i (j=1, \cdots, k)$.

This problem can be solved by common ILP solver. However,
computationally ILP is one of the most difficult combinatorial
optimization problems and the runtime is unaccepted even if the
problem size is small. In the following subsections, we will explain
how to transform this problem to a convex cost dual network flow
problem.

\subsection{Formulation Simplification}
Constraint (\ref{model1:t>s}) make problem (\ref{eq:model1}) too
complex to solve by network flow-based algorithm. First we consider
a more simple formulation (\ref{eq:model2}), which removes
constraint (\ref{model1:t>s}). To compensate the lose of accuracy,
we add penalty function $P(t_{ij})$ in objective function.
\begin{align}
    \label{eq:model2}
    \textrm{min}  & \ \  \sum_{i\in V} P(\bar s_i) + \sum_{(i, j)\in E} P(t_{ij})     \tag{$\uppercase\expandafter{\romannumeral3}$}  \\
    \textrm{s.t}. & \ \  (IIa) - (IIg)                                {} \notag\\
                  & \ \  t_{ij} \ge - c \cdot w_{ij}, \ \ \  \forall (i,j)\in E {} \notag
\end{align}
where $P(t_{ij}) = P(\bar s_j)/k$, and $k$ is a coefficient. Here we
set $k = \sum_i (1 - w_{ij})$\footnote{We suppose for each $(i,
j)\in E$, $w_{ij}$ is $0-1$ variable.}.

Given solution of problem (\ref{eq:model2}) $\bar s_i(i=1,\dots, m)$
and $t_{ij}(\forall (i,j)\in E)$, we propose a heuristic method to
generate solution of problem (\ref{eq:model1}).
\begin{eqnarray}
    \label{eq:transolution}
    t_{i,j}\ge & \bar s_j - c \cdot w_{ij} \Rightarrow \bar s_j  = \textrm{min}(t_{ij}+c \cdot w_{ij}), \forall i \in FI(j)
\end{eqnarray}

We denote the $\bar s_j$ got in (\ref{eq:transolution}) as $\bar
s_j(\Omega)$ and $\bar s_j$ got from problem (\ref{eq:model2}) as
$\bar s_j(\Theta)$, then we can get $\bar s_j$ in problem
(\ref{eq:model1}) as follows:
\begin{eqnarray}
    \bar s_j & = & \textrm{min} [\bar s_j(\Omega), \bar s_j(\Theta)] {}\notag \\
        & = & \textrm{min} [\textrm{min}(t_{ij}+c \cdot w_{ij}),
        \bar s_j(\Theta)], \quad \forall i \in FI(j)
\end{eqnarray}

By now we have build the connection between solution of problem
(\ref{eq:model1}) and problem (\ref{eq:model2}). After we calculate
the solution of (\ref{eq:model2}), we can then get the solution of
(\ref{eq:model1}). In the next subsection, we will prove problem
(\ref{eq:model2}) can be transformed to convex cost dual network
flow problem.

\subsection{Remove Redundant Constraint}
In this subsection we will prove that without loss of optimality,
problem (\ref{eq:model2}) can remove constraint $\bar R_i - \bar r_i
\le T$.

Let $s_i^*$ denote the value of $s_{i}$ for which $P(\bar s_i)$ is
minimum. In case there are multiple values for which $P(\bar s_i)$
is minimum, the minimum value will be chosen. Let us define the
function $Q(\bar s_i)$ in the following manner:
\begin{equation}
    \label{eq:Q}
    Q(\bar s_{i}) = \left\{
      \begin{array}{lll}
        P(\bar s_i^*) & if & \bar s_i \le s_i^* \\
        P(\bar s_i)   & if & \bar s_i > s_i^*
      \end{array}
    \right.
\end{equation}

Now consider the following problem (\ref{eq:model2'}), which
replaces ($IIa$) and ($IIb$) by $\bar R_i - \bar r_i = \bar s_i$:
\begin{align}
    \label{eq:model2'}
    \textrm{min}  & \ \  \sum_{i\in V} Q(\bar s_i) + \sum_{(i, j)\in E} P(t_{ij})&      \tag{$III'$}  \\
    \textrm{s.t}. & \ \  (IIc)-(IIg)          {}\notag\\
                  & \ \  \bar R_i - \bar r_i = \bar s_i            &\forall i \in V        {}\notag\\
                  & \ \  t_{ij} \ge - T \cdot w_{ij}    &\forall (i, j) \in E {} \notag
\end{align}

\begin{theorem}
    \label{theo:1}
    For every optimal solution ($\bar R, \bar r, \bar s$) of problem (\ref{eq:model2}), there is
    an optimal solution ($\bar R, \bar r, \hat s$) of problem (\ref{eq:model2'}),
    and the converse also holds.
\end{theorem}

\begin{proof}
Consider an optimal solution ($\bar R, \bar r, \bar s$) of
(\ref{eq:model2}), we show how to construct an optimal solution
($\bar R, \bar r, \hat s$) of (\ref{eq:model2'}) with the same cost.
There are two cases to consider:

\textbf{Case 1: $\bar R_i - \bar r_i \ge s_i^*$}. It follows from
(\ref{model1:Rr>s}) and the convexity of $P(\bar s_i)$ that $\hat
s_i = s_i^*$. In this case, we set $\hat s = \bar R_i - \bar r_i$.
It follows from (\ref{eq:Q}) that $P(\bar s_i) = Q(\hat s_i)$.

\textbf{Case 2: $\bar R_i - \bar r_i < s_i^*$}. Similar to case 1,
we can get $\bar s_i = \bar R_i - \bar r_i$. In this case, we set
$\hat s_i = \bar R_i - \bar r_i$. It follows from (\ref{eq:Q}) that
$P(\bar s_i) = Q(\hat s_i)$.

Similarly, it can be shown that if ($\hat R, \hat r, \hat s$) is an
optimal solution of (\ref{eq:model2'}), then the solution ($\hat R,
\hat r, \bar s$) constructed in the following manner is an optimal
solution of (\ref{eq:model2}): $\bar s_i = $max$\{s_i^*, \hat
s_i\}$.

\end{proof}

\begin{theorem}
 \label{theo:2}
The constraint $\bar R_i - \bar r_i \le T$ in problem
(\ref{eq:model2}) can be removed.
\end{theorem}

\begin{proof}
By Theorem \ref{theo:1}, we can transform each constraint in
(\ref{model1:Rr>s}) to an equality constraint. In other words, $\bar
R_i - \bar r_i = \bar s_i$. Because constraint (\ref{model1:s}) ($0
\le \bar s_i \le T$), $\bar R_i - \bar r_i \le T$. So we can remove
constraint $\bar R_i - \bar r_i \le T$.

\end{proof}

\subsection{Transformation to Primal Network Flow Problem}
To further simplify problem (\ref{eq:model2}), we transform $G(V,
E)$ into $\bar G(\bar V, \bar E)$ in such a way that each vertex $i
\in V$ is split into two vertex $\bar r_i$ and $\bar R_i$. So
constraints (\ref{model1:Rr>s}) (\ref{model1:RR>tcw}) and
(\ref{model1:rr<-cw}) can be transformed to the connection
relationship in $\bar E$. $\bar V = \{\bar r_1, \bar R_1, \dots,
\bar r_m, \bar R_m\}$. $\bar E = \bar E_1 \cup \bar E_2 \cup \bar
E_3$, where $\bar E_1$ include edges $(\bar r_i, \bar R_i)$, $\bar
E_2$ include edges $(\bar R_i, \bar R_j)$ and edges $(\bar r_i, \bar
r_j)$ belong to $\bar E_3$. Fig. (\ref{fig:DAG1}) illustrates a
simple DAG $G$ representing a synchronous sequential circuit, and
the transformed DAG $\bar G$ of $G$ is illustrated in Fig.
(\ref{fig:DAG2}).

\begin{figure} [tb]
  \centering%
  \subfloat[]{%
    \label{fig:DAG1}
    \includegraphics[width=0.20\textwidth]{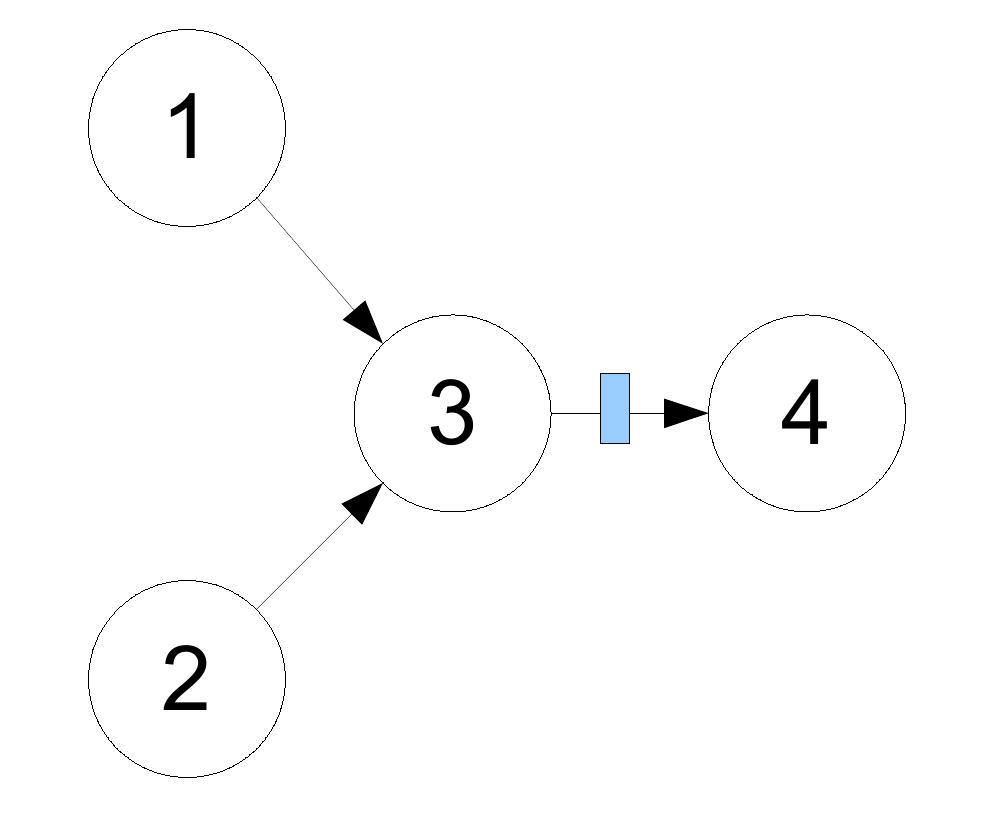}}\hspace{0.1em}%
  \subfloat[]{%
    \label{fig:DAG2}
    \includegraphics[width=0.24\textwidth]{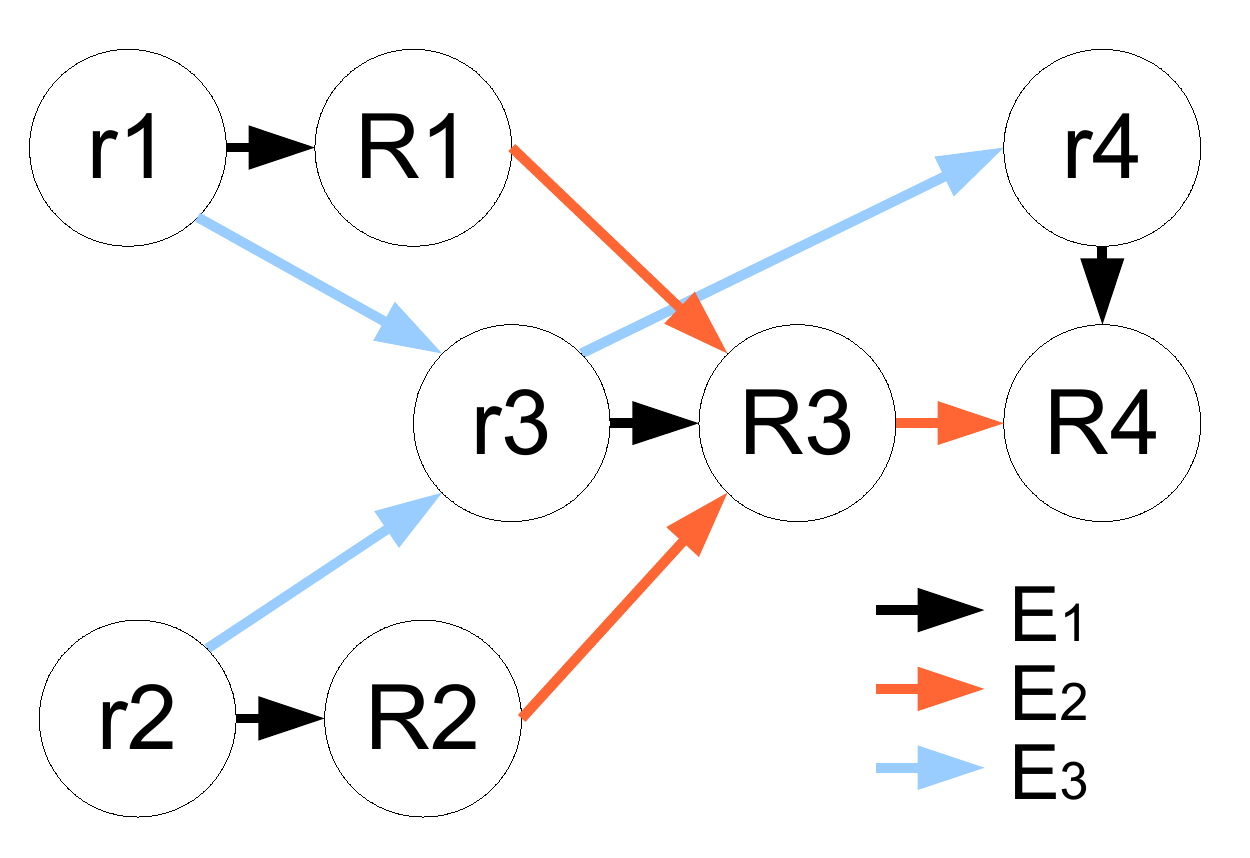}}
  \caption{~(a)The DAG $G$ representing a synchronous sequential circuit. ~(b)The transformed DAG $\bar G$ of $G$.}
  \label{fig:DAGs}
\end{figure}

\begin{figure*} [tb]
  \centering%
  \subfloat[]{%
    \label{fig:PijE1}
    \includegraphics[width=0.24\textwidth]{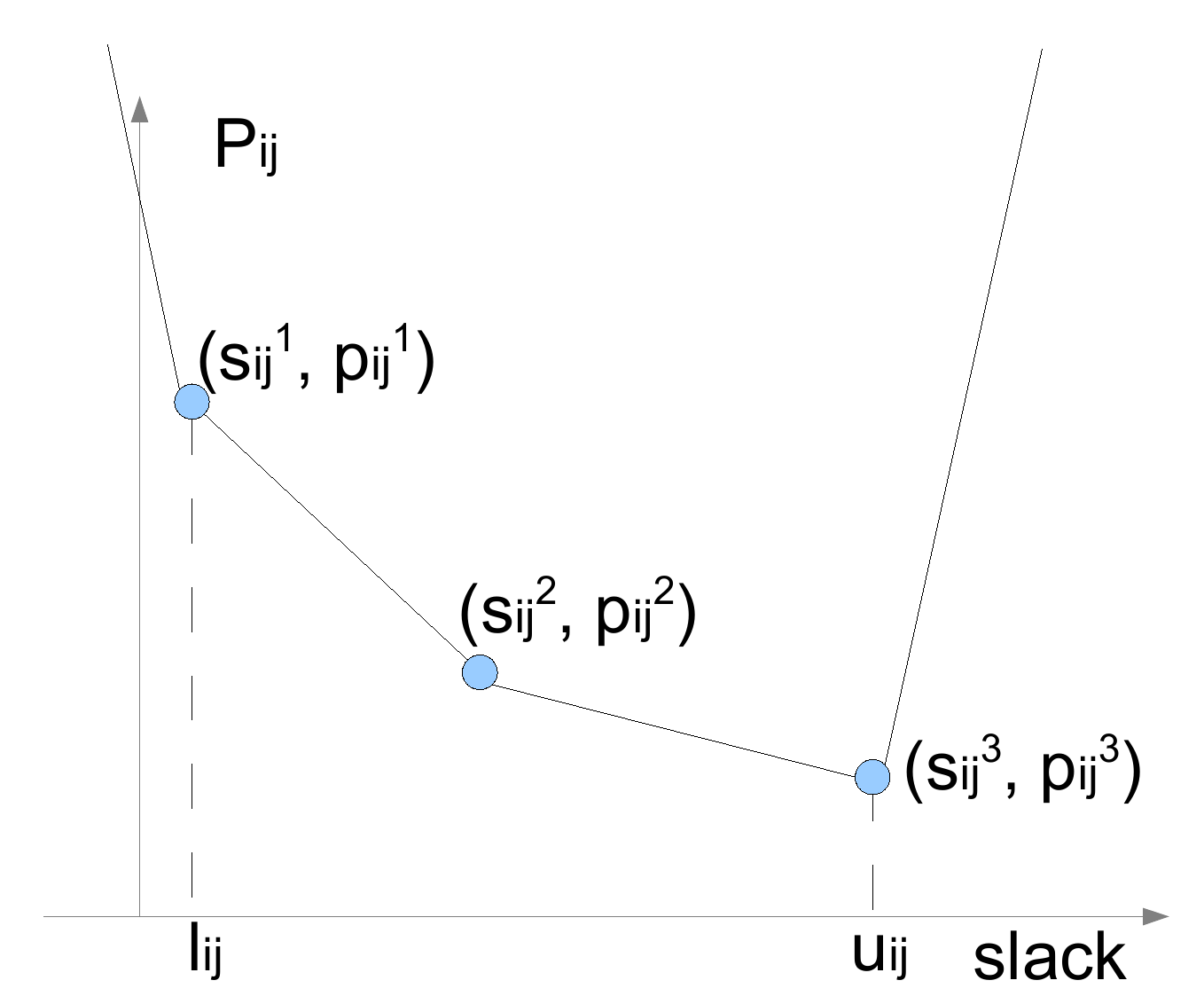}}\hspace{0.3em}%
  \subfloat[]{%
    \label{fig:PijE2}
    \includegraphics[width=0.28\textwidth]{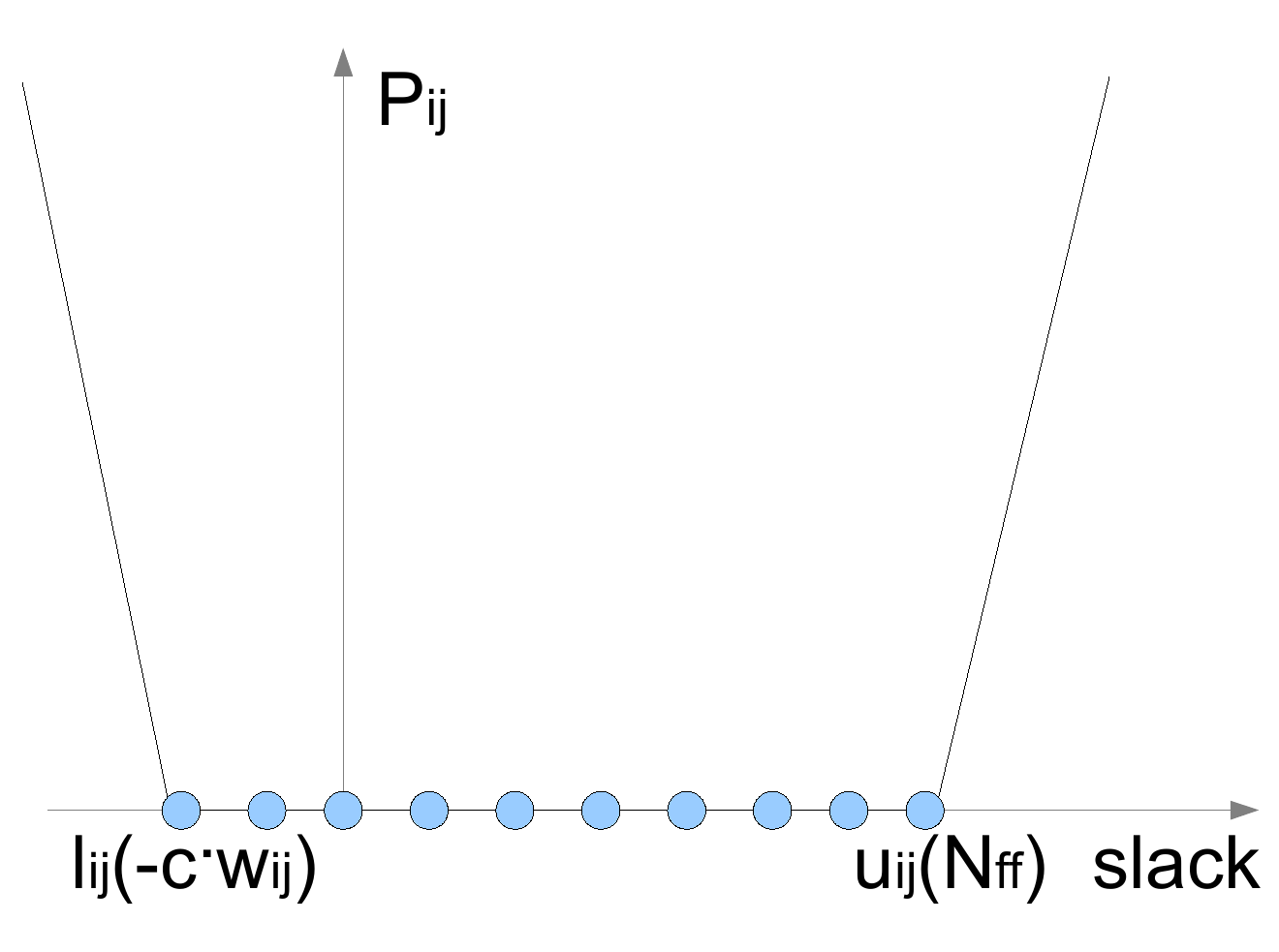}}\hspace{0.3em}
  \subfloat[]{%
    \label{fig:PijE4}
    \includegraphics[width=0.26\textwidth]{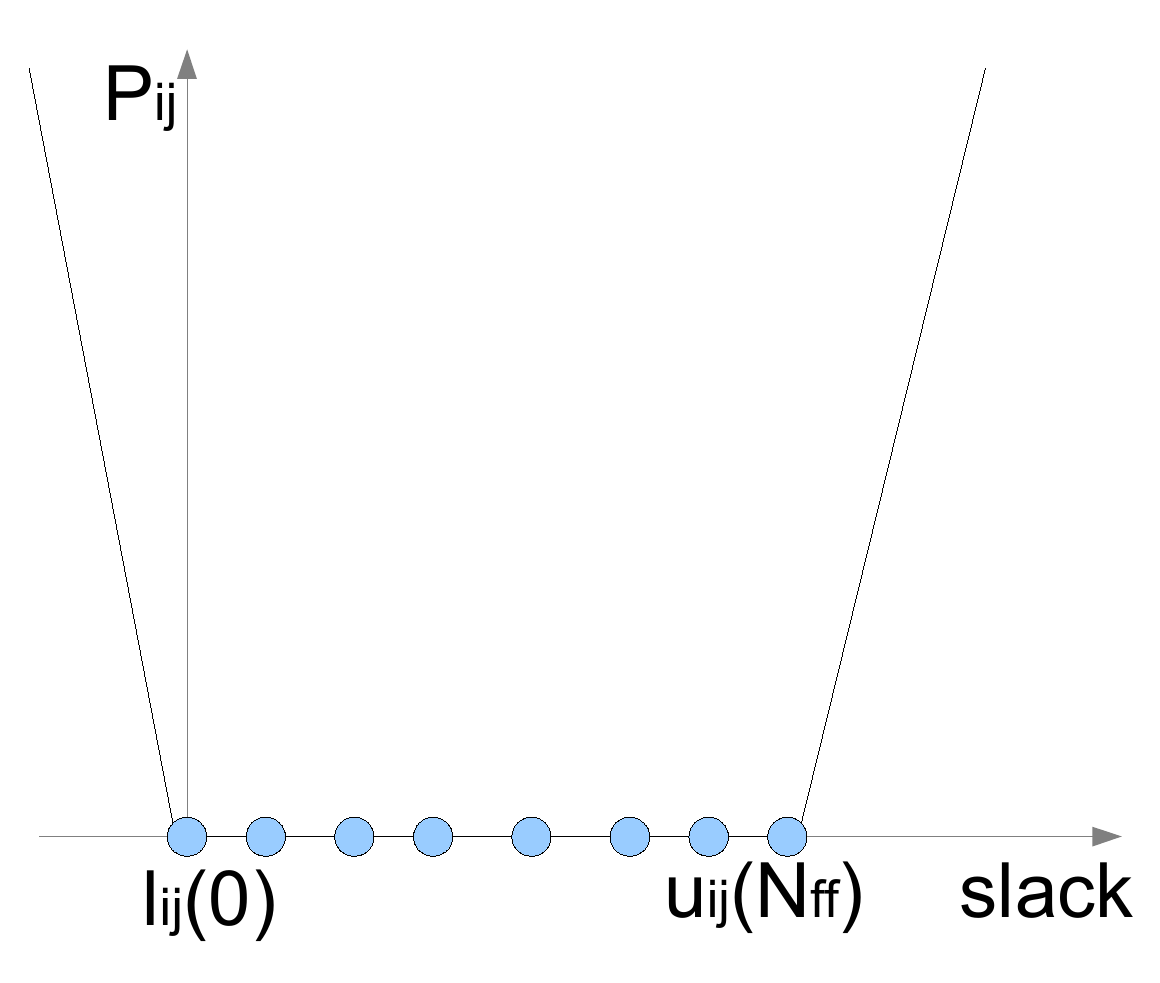}}
  \caption{~The Power-Slack Curve of ~(a)an edge $(i, j) \in E_1 \cup E_2$, here we assume $w_{ij}=0$; ~(b)an edge $(i, j) \in E_3$;~(c)an edge $(i, j) \in
  E_4$.}
  \label{fig:Pij}
\end{figure*}

Now the problem formulation can be simplified as follows:
\begin{align}
    \label{eq:model4}
    \textrm{min}  &     \ \ \ \  \sum_{(i, j) \in \bar E} P(s_{ij})& \tag{$IV$}\\
    \textrm{s.t}.\ \ & \mu_j - \mu_i \ge  s_{ij}                   &\forall (i, j) \in \bar E  \label{eq:cl3a}\tag{$IVa$}\\
            & 0 \le \mu_i \le \bar N_{ff}                    &\forall i \in \bar V       \label{eq:cl3b}\tag{$IVb$}\\
            & l_{ij} \le s_{ij} \le u_{ij}             &\forall (i, j) \in \bar E  \label{eq:cl3c}\tag{$IVc$}
\end{align}
where $s_{ij}$ represents slack assigned to edge from node $i$ to
$j$. For each edge $e(i,j) \in E_1$, if $i=\bar r_p$ and $j=\bar
R_p$, then $s_{ij}=\bar s_{p}$, and $l_{ij} = \bar s_{p}^1$ and
$u_{ij} = \bar s_{p}^k$. For each edge $e(i,j) \in E_2$, $s_{ij} =
\bar s_j - T \cdot w_{ij}$, then $l_{ij} = \bar s_j^1 - T \cdot
w_{ij}$ and $u_{ij} = \bar s_j^k - T \cdot w_{ij}$. For each edge
$e(i,j) \in E_3$, $l_{ij} = - T \cdot w_{ij}$ and $u_{ij} = \bar
N_{ff}$. An example Power-Slack Curve of an edge in $E_1 \cup E_2$
and that of an edge in $E_3$ are illustrated in Fig.
(\ref{fig:PijE1}) and Fig. (\ref{fig:PijE2}), respectively.

We then further eliminate constraints (\ref{eq:cl3b}) and
(\ref{eq:cl3c}). First of all, $P(s_{ij})$ can be modified to
eliminate the bounds on $\bar s_i$ as follows.
\begin{equation}
    \bar P(s_{ij}) = \left\{
      \begin{array}{ll}
        P(u_{ij}) + M(s_{ij} - u_{ij} )& \bar s_{ij} > u_{ij}\\
        P(s_{ij})                    & 0 \le \bar s_i \le T \\
        P(l_{ij}) - M(s_{ij} - l_{ij} ) & \bar s_{ij} < l_{ij}
      \end{array}
    \right.
\end{equation}
where $M$ is a sufficiently large number such that $\bar P(s_{ij})$
is still a convex function.

Similarly, the bounds on $\mu_i$ can also be eliminated by adding
into objective a convex cost function $B(\mu_i)$ defined as follows.
\begin{equation}
    B(\mu_i) = \left\{
      \begin{array}{lll}
        M \cdot(\mu_i - \bar N_{ff})           &if &  \mu_i > \bar N_{ff}\\
        0                                      &if & 0 \le \bar \mu_i \le \bar N_{ff} \\
        -M \cdot \mu_i                         &if &  \mu_i < 0
      \end{array}
    \right.
\end{equation}

After the above simplifications, problem (\ref{eq:model4}) can be
transformed to problem (\ref{eq:model5}):
\begin{align}
    \label{eq:model5}
    \textrm{min}  &     \ \ \ \  \sum_{(i, j) \in \bar E} \bar P(s_{ij}) + \sum_{i\in \bar V}B(\mu_i)& \tag{$V$}\\
    \textrm{s.t}.\ \ & \mu_j - \mu_i \ge  s_{ij}                   &\forall (i, j) \in \bar E \notag
\end{align}

\subsection{Problem Transformation by Lagrangian Relaxation}
Using Lagrangian relaxation to eliminate constraint in problem
(\ref{eq:model5}), get the Lagrangian sub-problem:
\begin{eqnarray}
    \label{eq:lgsubproblem}
    L(\vec x) & = & \sum_{e(i,j)\in{\bar{E}}} \bar P(s_{ij}) + \sum_{i\in \bar
    V}B_i(\mu_i)\nonumber\\
              &   & -\sum_{e(i,j)\in \bar E}(\mu_j-\mu_i - s_{ij})x_{ij}
\end{eqnarray}

It is easy to show that
\begin{equation}
  \sum_{e(i,j)\in \bar E}(u_i-u_j)x_{ij}=\sum_{i \in \bar V}x_{0i}\times \mu_i
\end{equation}
where
\begin{equation}
x_{0i}=\sum_{j:e(i,j)\in \bar E}x_{ij}-\sum_{j:e(j,i)\in \bar
E}x_{ji}
 , \forall i \in V
\end{equation}
Lagrangian subproblem (\ref{eq:lgsubproblem}) can be restated as
follows:
\begin{equation}\label{eq:lgsubproblem2}
  L(\vec{x})=\textrm{min} \sum_{e(i,j)\in \bar E}[P(s_{ij}) +x_{ij}s_{ij}]+\sum_{i\in
  \bar V}[B_i(\mu_i)+x_{0i}\mu_i]
\end{equation}

A start node $v_0$ is added to $\bar V$, $v_0$ interconnects all
other nodes in $\bar V$. We set $s_{0i}=\mu_i, l_{0i}=0, u_{0i}=\bar
N_{ff}$. So $V = \{v_0\} \cup \bar V$. The new edges are denoted as
$E_4$, $E = \bar E \cup E_4$. The Power-Slack curve of an edge $(i,
j)\in E_4$ is illustrated in Fig. (\ref{fig:PijE4}). So we can
transform $L(\vec x)$ as formulation (\ref{eq:lx}).

\begin{align}
    \label{eq:lx}
                     &  \ \  L(\vec{x})=\textrm{min} \sum_{e(i,j)\in E}[P_{ij}(s_{ij})+x_{ij}s_{ij}]  \\
    \textrm{s.t}.    &  \sum_{j:e(i,j)\in E}x_{ij}-\sum_{j:e(j,i)\in E}x_{ji}=0 \quad \forall i \in V \notag\\
                     &  x_{ij} \ge 0 \qquad \qquad \forall (i, j)\in E_1 \cup E_2 \cup E_3 \notag
\end{align}

\subsection{Convex Cost-scaling Approach}
We define function $H_{ij}(x_{ij})$ for each $e(i,j)\in E$ as
follows:
\begin{align}
  H_{ij}(x_{ij})=\textrm{min}_{s_{ij}}\{P_{ij}(s_{ij})+x_{ij}s_{ij}\}
\end{align}

For the $e(i,j)\in E_1$, because the function $H_{ij}(x_{ij})$ is a
piecewise linear concave function of $x_{ij}$, and $\forall e(i,
j)\in E_1$, then $H_{ij}(x_{ij})$ is described in the following
manner \cite{flow_2001Ahuja}: {\setlength\arraycolsep{2pt}
\begin{eqnarray}
H_{ij}(x_{ij}) & =& \left \{
                 \begin{array}{ll}
                   P_{ij}(s_{ij}^k)+s_{ij}^kx_{ij} & 0 \leq x_{ij}\leq b_{ij}(k)\\
                   \dots \\
                   P_{ij}(s_{ij}^q)+s_{ij}^qx_{ij} & b_{ij}(q+1) \leq x_{ij}\leq b_{ij}(q)\\
                   \dots \\
                   P_{ij}(s_{ij}^1)+s_{ij}^1x_{ij} & k \leq x_{ij}\\
                 \end{array}
               \right. \notag
\end{eqnarray}}
where
$b_{ij}(q)=\frac{P_{ij}(s_{ij}^{q-1})-P_{ij}(s_{ij}^q)}{s_{ij}^q-s_{ij}^{q-1}}$.

For the $e(i,j)\in E_1$, similar to $E_2$, then $H_{ij}(x_{ij}) =$
{\setlength\arraycolsep{2pt}
\begin{eqnarray}
H_{ij}(x_{ij}) & =& \left \{
                 \begin{array}{ll}
                   P_{ij}(t_{ij}^k)+t_{ij}^kx_{ij} & 0 \leq x_{ij}\leq b_{ij}(k)\\
                   \dots \\
                   P_{ij}(t_{ij}^q)+t_{ij}^qx_{ij} & b_{ij}(q+1) \leq x_{ij}\leq b_{ij}(q)\\
                   \dots \\
                   P_{ij}(t_{ij}^1)+t_{ij}^1x_{ij} & k \leq x_{ij}\\
                 \end{array}
               \right. \notag
\end{eqnarray}
} where
$b_{ij}(q)=\frac{P_{ij}(t_{ij}^{q-1})-P_{ij}(t_{ij}^q)}{t_{ij}^q-t_{ij}^{q-1}}$,
and $t_{ij}^q = s_{ij}^q - T \cdot w_{ij}$.

For the $e(i,j)\in E_3$, because $P_{ij}(s_{ij})=0$,
\begin{align}
    H_{ij}(x_{ij})=min_{s_{ij}}(s_{ij} x_{i,j})= - T \cdot w_{ij}\cdot x_{i,j},  x_{ij} \ge
    0 \notag
\end{align}

For the $e(i,j)\in E_4$, the variable $x_{i,j}$ is not a Lagrangian
multiplier, and it is bounded by $-M \le x_{ij} \le M$.
\begin{equation}
    H_{ij}(x_{ij}) = \left\{
      \begin{array}{lll}
        0                                      & 0 \le x_{ij} \le M \\
        \bar N_{ff} \cdot x_{ij}               & -M \le x_{ij} \le 0
      \end{array}
    \right. \notag
\end{equation}

Note that these functions $H_{ij}(x_{ij})$ are all concave. We
define $C_{ij}(x_{ij}) = - H_{ij}(x_{ij})$, so that $C_{ij}(x_{ij})$
is a piecewise linear convex function. Then we can subsequently
propose problem (\ref{eq:flowproblem}) as follows:

\begin{align}
    \label{eq:flowproblem}
                     &  \ \  L(\vec{x})=\textrm{min} \sum_{e(i,j)\in E}C_{ij}(x_{ij})  \tag{$\uppercase\expandafter{\romannumeral6}$}\\
    \textrm{s.t}.    &  \sum_{j:e(i,j)\in E}x_{ij}-\sum_{j:e(j,i)\in E}x_{ji}=0 \quad \forall i \in V \notag\\
                     &  0 \le x_{ij} \le M \qquad \qquad \forall (i, j)\in E_1 \cup E_2 \cup E_3 \notag \\
                     &  -M \le x_{ij} \le M \qquad \qquad \forall (i, j)\in E_4 \notag
\end{align}

To transform the problem into a minimum cost flow problem, we
construct an expanded network $G' = (V', E')$. There are four kinds
of edges to consider:
\begin{flushleft}
\begin{itemize}
\item
  $e(i,j)$ in E1: we introduce $k$ edges in $G'$, and the costs of these
  edges are: $-s_{ij}^k, -s_{ij}^{k-1}, \dots -s_{ij}^1$; upper capacities:
  $b_{ij}(k), b_{ij}(k-1)-b_{ij}(k), b_{ij}(k-2)-b_{ij}(k-1), \dots
  M-b_{ij}(2)$, where $M$ is a huge coefficient; lower capacities are all 0.
\item
  $e(i,j)$ in E2: we introduce $k$ edges in $G'$, and the costs of these
  edges are: $-t_{ij}^k, -t_{ij}^{k-1}, \dots -t_{ij}^1$; upper capacities:
  $b_{ij}(k), b_{ij}(k-1)-b_{ij}(k), b_{ij}(k-2)-b_{ij}(k-1), \dots
  M-b_{ij}(2)$, where $M$ is a huge coefficient; lower capacities are all 0.
\item
  $e(i,j)$ in E3: cost, lower and upper capacity is ($c \cdot w_{ij}, 0, M$).
\item
  $e(i,j)$ in E4: two edges are introduced in $G'$, one with cost, lower
  and upper capacity as ($\bar N_{ff}, -M, 0$), another is ($0,0,M$).
\end{itemize}
\end{flushleft}

Using the cost-scaling algorithm \cite{book01:flow}, we can solve
the minimum cost flow problem in $G'$. For the given optimal flow
$x^*$, we construct residual network $G(x^*)$ and solve a shortest
path problem to determine shortest path distance $d(i)$ from node
$s$ to every other node. By implying that $\mu(i) = d(i)$ and
$s_{ij} = \mu(i) - \mu(j)$ for each $e(i,j) \in E_1 \cup E2$, we can
finally solve problem (\ref{eq:model2}).

\section{Experimental Results}
\label{sec:result}

We implemented our algorithm in the C++ programming language and
executed on a Linux machine with eight 3.0GHz CPU and 6GB Memory. 19
cases from the ISCAS89 benchmarks are tested, and the name, number
of gates, number of signal passes, the maximum number of gate
output/inputs, and the minimum period for each case are given in
Table \ref{table:benchmarks}. We used four discrete slack levels for each gate
as $\{0, 10, 20, 33\}$. Energy consumption of the gates with slack level
scaling were found from model in \cite{slack_ISQED08Ma}.

\begin{table}[tb]
\renewcommand{\arraystretch}{1}
\centering \caption{Characteristics of Test Cases}
\label{table:benchmarks}
\begin{tabular}{|c|ccccc|}
 \hline \hline
    Case Name   & Gate \#   & Edges \#  & Max Output    & Max Inputs   & Tmin\\
 \hline
    s27.test    & 11    & 19    & 4     & 2     & 20\\
    s208.1.test & 105   & 182   & 19    & 4     & 28\\
    s298.test   & 120   & 250   & 13    & 6     & 24\\
    s382.test   & 159   & 312   & 21    & 6     & 44\\
    s386.test   & 160   & 354   & 36    & 7     & 64\\
    s344.test   & 161   & 280   & 12    & 11    & 46\\
    s349.test   & 162 &284 &12 &11  &46\\
    s444.test   & 182 &358 &22 &6   &46\\
    s526.test   & 194 &451 &13 &6   &42\\
    s526n.test  & 195 &451 &13 &6   &42\\
    s510.test   & 212 &431 &28 &7   &42\\
    s420.1.test & 219 &384 &31 &4   &50\\
    s832.test   & 288 &788 &107 &19   &98\\
    s820.test   & 290 &776 &106 &19   &92\\
    s641.test   & 380 &563 &35 &24   &238\\
    s713.test   & 394 &614 &35 &23   &262\\
    s838.1.test & 447 &788 &55 &4   &80\\
    s1238.test  & 509 &1055 &192 &14   &110\\
    s1488.test  & 654 &1406 &56 &19   &166\\
 \hline \hline
\end{tabular}
\end{table}

\begin{table*}[tb]
 \renewcommand{\arraystretch}{1.2}
 \centering \caption{Comparisons with Optimal ILP and Previous Work \cite{retiming_ASPDAC10Liu}}
 \label{table:result}
\begin{tabular}{|c|c|ccc|ccc|ccc|}
 \hline \hline
 Benchmark & T & \multicolumn{3}{|c}{Power Consumption} & \multicolumn{3}{|c}{Total Slacks}
           &\multicolumn{3}{|c|}{Runtime(s)}\\
 & & Optimal ILP& \cite{retiming_ASPDAC10Liu}&  ours& Optimal ILP & \cite{retiming_ASPDAC10Liu} & ours &Optimal ILP& \cite{retiming_ASPDAC10Liu} &ours\\
 \hline
    s27.test    & 20    &800&824&850        & 40 & 40 & 30      &0.02&0.0& 0.0\\
    s208.1.test & 28    &3542&9118&4772     &1770&290&1988      &0.39&0.44& 0.06\\
    s298.test   & 24    &6498&8888&8010     &1330&660&1240      &0.78&0.69& 0.07\\
    s382.test   & 44    &6456&9038&9958     &3011&2071&1895     &$>$1000&10.56& 0.12\\
    s386.test   & 64    &8836&12870&9564    &2484&807&2324      &4.58&1.03& 0.1\\
    s344.test   & 46    &9876&11848&9894    &1855&1064&1760     &0.82&2.53& 0.09\\
    s349.test   & 46    &9938&12472&9894    &1852&912&1780      &0.79&4.49& 0.11\\
    s444.test   & 46    &8938&14032&11884   &2962&1025&1939     &$>$1000&12.04& 0.12\\
    s526.test   & 42    &7602&14106&11498   &3626&1307&2356     &42.57&1.67& 0.17\\
    s526n.test  & 42    &7752&11734&11548   &3616&2089&2366     &30.32&4.72& 0.17\\
    s510.test   & 42    &13976&17492&14846  &2237&937&2040      &$>$1000&1.62& 0.17\\
    s420.1.test & 50    &4574&17920&9224    &5906&1050&4466     &1.29&16.91& 0.14\\
    s832.test   & 98    &13652&14518&16274  &5175&4525&4171     &71.96&151.26&0.24\\
    s820.test   & 92    &13552&17694&16448  &5261&3493&4103     &68.98&13.18& 0.25\\
    s641.test   & 238   &13334&20408&14424  &7925&6067&7604     &2.24&92.97& 0.26\\
    s713.test   & 262   &13018&21228&14322  &8522&6363&8112     &2.27&121.1& 0.27\\
    s838.1.test & 80    &6004&18898&17556   &14048&9016&9912    &1.48&256.9& 0.4\\
    s1238.test  & 110   &6096&10444&8208    &16764&14635&15792  &0.23&448.6& 0.34\\
    s1488.test  & 166   &21292&23799&27836  &15313&14791&13024  &$>$1000&670.7& 0.53\\
    \hline
    Avg         & -   &9249.3&14070&11947.9 &5457.7&3744.3&4573.8   & - &95.3& 0.19\\
    Diff        & -     & 1 & +52\%& +29\%    & 1 & -31\% & -16\%   & - & 1 & 0.002\\
\hline \hline
\end{tabular}
\end{table*}

In the experiments, a min-period retiming algorithm \cite{retiming_ACM08Zhou} is first employed
to generate the minimum clock period $T$, which is listed in the 2nd column of TABLE \ref{table:result}.
Liu et al.'s \cite{retiming_ASPDAC10Liu} algorithm was implemented for comparison.
Note that algorithm in \cite{retiming_ASPDAC10Liu} can not directly solve discrete slack budgeting problem,
because if sensitive transitive closure graph is used, the timing constraints might be violated after slack budgeting
\cite{slack_DAC96Maji}. Therefore we use a transitive closure graph instead of sensitive transitive closure graph here.
To evaluate the accuracy of our algorithm, the ILP for achieving the optimal solution were also implemented
using an open source ILP solver CBC \cite{cbc}.

Table \ref{table:result} shows comparisons among optimal ILP, algorithm in
\cite{retiming_ASPDAC10Liu} and our algorithm. The column Power Consumption gives actual
power consumption of each circuit and less value means more power can be reduced.
Comparing with optimal solution, our algorithm increases 29\% power consumption while
\cite{retiming_ASPDAC10Liu} increases 52\%. Column Total Slack gives the sum of each gate's slack.
Comparing with optimal solution, our algorithm loses 16\% of slacks while \cite{retiming_ASPDAC10Liu} loses 31\%.
Note that power consumption is not proportional to the slack amount. As for benchmark s27.test, \cite{retiming_ASPDAC10Liu}
and optimal ILP get equal slack amount, but their power consumption is different.
Column Runtime compares the run time of each algorithm. From the results we can find that although optimal ILP can get
optimal solution, its runtime sometimes is unacceptable. Comparing with \cite{retiming_ASPDAC10Liu},
our algorithm can not only generate better design results, but also get nearly $500 \times$ speedup.

\section{Conclusion}
\label{sec:conclusion}

In this paper we have showed that the
retiming and slack budgeting problem can be simultaneously solved by
formulating the problem to a convex cost dual network flow problem.
Both the theoretical analysis and experimental results show the
efficiency of our approach which can not only reduce power consumption but also speedup previous work.

\bibliographystyle{IEEEtran}
\bibliography{../ref/retiming,../ref/Algorithm,../ref/MVS}

\end{document}